\documentclass[preprint,floatfix,superscriptaddress,11pt,nofootinbib]{revtex4-2}
\usepackage{float}
\makeatletter
\let\newfloat\newfloat@ltx
\makeatother
\usepackage{algorithm,algpseudocode,amssymb,amsthm,graphicx,mathtools}
\usepackage[T1]{fontenc}

\usepackage{color}
\usepackage{colortbl}
\usepackage{nicefrac}
\usepackage[hypertexnames=false,colorlinks,urlcolor=blue,citecolor=blue]{hyperref}
\usepackage{natbib}
\usepackage[capitalize]{cleveref}
\usepackage{dsfont}

\makeatletter
\DeclareRobustCommand{\cev}[1]{%
  \mathpalette\do@cev{#1}%
}
\makeatother

\newcommand{\eg}{\emph{e.g.,~}}
\newcommand{\ie}{\emph{i.e.,~}}
\newcommand{\etal}{\emph{et al.}}

\newcommand{\nn}{\nonumber}

\newcommand{\dg}{^\dagger}
\newcommand{\bra}[1]{\langle #1|}
\newcommand{\ket}[1]{|#1\rangle}
\newcommand{\tr}{{\rm tr}}
\newcommand{\expect}[1]{\langle #1\rangle}

\newcommand{\ketbra}[2]{\ket{#1}\bra{#2}}

\newcommand{\norm}[1]{\left \lVert #1 \right \rVert}

\newcommand{\bR}{{\bf R}}
\newcommand{\bmu}{{\boldsymbol{\mu}}}
\newcommand{\bsigma}{{\boldsymbol{\sigma}}}
\newcommand{\bSigma}{{\boldsymbol{\Sigma}}}
\newcommand{\bgamma}{{\boldsymbol{\gamma}}}
\newcommand{\bo}{{\bf p}}

\newtheorem{theorem}{Theorem}
\newtheorem{lemma}[theorem]{Lemma}

\newtheorem{manualtheoreminner}{Theorem}
\newenvironment{manualtheorem}[1]{%
  \IfBlankTF{#1}
    {\renewcommand{\themanualtheoreminner}{\unskip}}
    {\renewcommand\themanualtheoreminner{#1}}%
  \manualtheoreminner
}{\endmanualtheoreminner}

\makeatletter
\newcommand{\vast}{\bBigg@{4}}
\newcommand{\Vast}{\bBigg@{5}}
\makeatother

\newcommand{\snlca}{Quantum Algorithms and Applications Collaboratory, Sandia National Laboratories, Livermore, CA 94550, USA}
\newcommand{\snlnm}{Quantum Algorithms and Applications Collaboratory, Sandia National Laboratories, Albuquerque, NM, 87123, USA}
\newcommand{\unm}{Center for Quantum Information and Control (CQuIC), Department of Physics and Astronomy, University of New Mexico, Albuquerque, NM 87131, USA}
\newcommand{\nus}{Department of Computer Science, School of Computing, National University of Singapore, Singapore 117417, Republic of Singapore}

\newcommand{\papertitle}{Learning Gaussian optical states with quantum computers}

\begin{document}
\title{\papertitle}
\author{Spencer Dimitroff}
\email{sddimit@sandia.gov}
\affiliation{\unm}
\affiliation{\snlca}
\author{John Kallaugher}
\affiliation{\nus}
\author{Ashe Miller}
\affiliation{\snlnm}
\author{Mohan Sarovar}
\email{mnsarov@sandia.gov}
\affiliation{\snlca}
\date{\today }

\begin{abstract}
Recent results have established dramatic advantages in learning properties of quantum states when a quantum computer is available to process or jointly measure multiple copies of the unknown quantum state. Learning tasks can be accomplished with exponentially fewer copies of the state when compared to optimized classical learning strategies that are restricted to measuring one copy of the state at a time. While these results were established in abstract settings and for artificial learning tasks, they motivate the application of quantum computers to imaging and sensing of weak electromagnetic fields since these settings are ultimately concerned with the learning of unknown quantum states. In this work we apply these new results in quantum learning to the problem of learning Gaussian states of the electromagnetic field, which are germane since they describe most fields used in imaging and sensing. In order to connect with quantum learning theory, we consider the transduction of an $n$-mode Gaussian state into a register of qubits on a quantum computer followed by optimized measurements on these qubits to extract the parameters defining the original Gaussian state. We rigorously bound the number of copies of the Gaussian state required to achieve worst-case additive error in parameter estimates. The scaling of this bound with $n$ is exponentially better than \emph{na\"ive} strategies for characterizing Gaussian states and matches recently derived bounds for characterization of Gaussian states using continuous-variable (CV) classical shadows. In addition, our bound has a polynomially better dependence on the energy of the multimode Gaussian state compared to the CV shadows protocol.
\end{abstract}

\maketitle

\section{Introduction}
Quantum computers have the potential to be a game-changing technology for a variety of fields. While the impact of quantum computing on computational tasks such as electronic structure calculations is widely recognized, recent results have established that having a quantum computer can also radically impact the task of learning about physical states and processes \cite{Chen2021-qi,Aharonov2022-ru,Huang2022-tz,King2024-lu}. 
In this setting, the advantage that quantum computers present is not in the time complexity for a task, but rather in \emph{sample complexity}. For example, Refs. \cite{Chen2021-qi,Huang2022-tz} show that by using a quantum computer, learning of quantum states and processes can be done to fixed error with, at times, exponentially fewer copies of the state or uses of the quantum process, than any classical strategy. These quantum advantages typically rely on the problem having some structure such as the need to simultaneously resolve non-commuting observables that makes joint (\ie multicopy) measurements advantageous over single-copy, classical measurements. This need for such multicopy or nonlocal measurements naturally connects to quantum computation where these types of complex operations can be performed.

The \emph{quantum learning} results mentioned above were established in an artificial setting where the source of the unknown quantum states and processes were left open and for fairly contrived learning tasks. However, they strongly motivate the application of this quantum learning theory to imaging and sensing applications since many important tasks in these areas, especially in the weak-field regime, are concerned with learning unknown quantum states or processes. 

With this motivation, we recently formulated the concept of \emph{quantum computational imaging and sensing} (QCIS) \cite{Sarovar_2023}, which captures the idea of performing imaging and sensing tasks on weak electromagnetic fields by first transducing quantum information from these fields into registers of qubits and then performing a quantum computation on them. Applying this concept to the task of decoding classical coherent communication, we showed that one could obtain a quantum advantage by joint coherent decoding (enabled by a quantum computation) of multiple communication pulses \cite{crossmanQuantumComputerenabledReceivers2024}. This strategy results in a lower probability of decoding error than possible with any classical receiver, even in the presence of noise in the transduction and the quantum computation. See \cite{Delaney2022-dr,smith2025quantumprocessingassistedclassicalcommunications} for related work on the same application.

In this paper we consider the application of QCIS to the learning of Gaussian states of the electromagnetic (EM) field. This setting significantly broadens the range of applications of the concept since almost all imaging and sensing applications involve Gaussian states of the EM field. Gaussian states capture common sources of light, such as thermal, coherent, partially coherent and chaotic light, as well exotic states such as squeezed states. The task we consider is the complete characterization of the Gaussian state of multiple ($n$) modes of an optical field\footnote{In this paper we use \emph{optical field/state} as a shorthand for electromagnetic field/state. This does not imply an assumption that the state of interest is in the optical frequencies.}. 
As explained below, an $n$-mode Gaussian state is completely specified by $2n^2 + 3n$ real numbers.
Some imaging and sensing tasks might require less information (\eg only the $2n$ mean values of the quadratures of the $n$ modes) but the task we consider subsumes any imaging and sensing task because if the entire $n$ mode state is characterized, any information encoded in it can be computed. However, we expect that if an application requires significantly less information about the state than all $2n^2 + 3n$ degrees of freedom, QCIS will present minimal advantage since classical measurement schemes could be optimized to target the few parameters of interest. We note that there is no dearth of applications where the whole Gaussian state needs to be characterized (\eg see the satellite imaging application studied in Refs. \cite{koseQuantumenhancedPassiveRemote2022,koseSuperresolutionImagingMultiparameter2023}).

Our QCIS scheme for learning $n$-mode Gaussian states proceeds by independently transducing each mode of the state into a qubit through the Jaynes-Cummings (JC) interaction. Then we show that by measuring one- and two-qubit Pauli observables on the $n$-qubit register, one can recover estimates of all Gaussian moments of the original optical state. The main result of the paper is summarized by the following theorem:

\begin{theorem}\label{main_result}
With probability $1-\delta$, optical-to-qubit transduction followed by one- and two-qubit Pauli measurements can estimate each of the $M=2n^2+3n$ means and covariances of an $n$-mode optical state to additive error $\epsilon$ using $T=\mathcal{O}\Big{(}\frac{E_{\rm max}^2 \log(n/\delta) \lceil \log n \rceil}{\epsilon^2}\Big{)}$ copies, where $E_{\rm max}$ is the maximum energy per mode.
\end{theorem}

The logarithmic dependence of the sample complexity on the number of modes, $n$, is remarkable. Many conventional optical strategies to fully characterize an $n$-mode Gaussian state require ${\rm poly}(n)$ copies of the state (\eg \cite{Kumar2020-mc}). Recent formulations of continuous-variable classical shadows \cite{gandhariPrecisionBoundsContinuousVariable2024,beckerClassicalShadowTomography2024} can also characterize an $n$-mode Gaussian state with $O(\log n)$ copies and conventional optical measurements, however, as we show in \cref{sec:CV shadows} their dependence on $E_{\rm \max}$ is polynomially worse. To achieve the sample complexity in \cref{main_result} we control the duration of the JC interaction to be in the perturbative regime, use classical shadows \cite{huangPredictingManyProperties2020,Huang2021-xl,huangProvablyEfficientMachine2022} to efficiently measure the required Pauli observables, and develop an iterative algorithm to recover estimates of the Gaussian moments with controlled error from the Pauli measurements. 

There have been several recent publications related to the central themes of this work, and we summarize these here. Firstly, several authors have considered the problem of learning and tomography of continuous-variable (bosonic) states \cite{gandhariPrecisionBoundsContinuousVariable2024,beckerClassicalShadowTomography2024,meleLearningQuantumStates2025,bittelOptimalEstimatesTrace2025,bittelEnergyindependentTomographyGaussian2025}. We mentioned the CV classical shadows results \cite{gandhariPrecisionBoundsContinuousVariable2024,beckerClassicalShadowTomography2024} above and will have a direct comparison of our work and this work in \cref{sec:CV shadows}. The other stream of recent work \cite{meleLearningQuantumStates2025,bittelOptimalEstimatesTrace2025,bittelEnergyindependentTomographyGaussian2025} considers the sample complexity of learning energy-constrained Gaussian states using optical measurements, but seeks to bound the error of the learned density matrix in terms of the trace distance. This is a much more stringent goal than ours, namely, to learn all Gaussian moments to some precision, which is more relevant to imaging and sensing applications. Our less stringent requirement explains the smaller sample complexity for the task considered in this paper when compared to Refs. \cite{meleLearningQuantumStates2025,bittelOptimalEstimatesTrace2025,bittelEnergyindependentTomographyGaussian2025}. 

Turning to work focused on applications of quantum computing to sensing and imaging, we first refer the reader to the review paper by Khan \etal\ \cite{khan2025quantumcomputationalsensingadvantage}, which covers many recent results in this area. 
We note that none of these results consider the problem we study in this paper of learning Gaussian states of the EM field explicitly, and in fact, most assume that the signal of interest is already encoded in the state of some qubits. Another direction of work has focused on increasing the resolution of astrophysical imaging, motivated by the seminal work by Gottesman \etal\ \cite{PhysRevLett.109.070503} that showed that the baseline of optical telescopes could be increased using quantum repeaters and teleportation with distributed entangled photons. Recent improvements of this idea reduce the amount of entanglement distribution needed by instead having small quantum computers at each telescope that perform coherent computations to compress the quantum information transduced from light \cite{Khabiboulline2019-gc,Khabiboulline2019-sq}. The state of the EM field measured in these works is thermal and therefore Gaussian. However, the problem they consider is more akin to single or few parameter estimation and not learning of the full Gaussian state. 
Finally, recent work on enhanced optical imaging by quantum computation by Mokeev \etal\ \cite{Mokeev2026-ia} is the most closely aligned with the work we present here since they consider learning properties of a multimode optical field by using coherent processing with a quantum computer, although the input state considered is non-Gaussian.

The rest of the paper is organized as follows. In \cref{sec:learning} we formally introduce the problem. In \cref{sec:transduction} we examine the JC transduction process as applied to carefully chosen initial states. We show how to recover all Gaussian parameters from low-weight Pauli observables of the transduced states, provided that higher-order components of the transduction map (that is, the order $>2$ terms in the Taylor expansion of the unitary time evolution) can be neglected. In \cref{sec:complexity} we use this observation to develop an iterative algorithm for recovering these parameters by performing increasingly accurate corrections for the aforementioned non-linear component of the transduction, and bound the sample complexity of this algorithm. \cref{sec:CV shadows} compares this sample complexity against the sample complexity of Gaussian state characterization using CV classical shadows. \cref{sec:conclusion} concludes with a discussion of our findings, and opportunities for extensions and future work.

\section{Gaussian state learning}\label{sec:learning}
Consider an $n$-mode continuous-variable bosonic state, $\rho_f$, which has $2n$ quadrature degrees of freedom, $\bR =(Q_1,P_1,...,Q_n,P_n)^T$. A Gaussian state \cite{ferraro2005gaussianstatescontinuousvariable,RevModPhys.84.621,adessoContinuousVariableQuantum2014} is completely determined by the first- and second-order moments of these quadratures, or equivalently, the means, ${\bmu}_i=\tr(\rho_f{\bf R}_i)$, and covariance matrix elements ${\bf\Sigma}_{jk}=\tr(\rho_f\{{\bf R}_j-\bmu_j,{\bf R}_k-\bmu_k\})/2$. In the following, we often collect the $2n^2 + n$ unique covariance matrix elements in a column vector $\bsigma$ and group them with the means to form the $2n^2+3n$-dimensional vector ${\bgamma} = \left(\begin{array}{c}
         \bmu \\
         \bsigma 
    \end{array}\right)$ that contains all parameters for the $n$-mode Gaussian state. We assume the second-order moments collected in $\bsigma$ are not the centralized moments that define $\bSigma$, but the non-centralized second-order moments, $\tr(\rho_f{\{{\bf R}_j,{\bf R}_k}\})/2$.

For an arbitrary multimode optical state $\rho_f$, we set the energy constraint of a single mode as
\begin{equation}
    \tr(\rho E_j)=\frac{1}{2}\tr(\rho Q_j^2)+\frac{1}{2}\tr(\rho P_j^2)\leq E_{\rm max}
\end{equation}
such that the total energy is bounded by $nE_{\rm max}$. Since both terms $\tr(\rho Q_j^2)$ and $\tr(\rho P_j^2)$ are non-negative, $\tr(\rho {\bf R}_j^2)\leq 2E_{\rm max}$, for all $j$. This implies via Cauchy-Schwarz that
\begin{equation} \label{covariance_bound}
    |\tr(\rho {\bf R}_j {\bf R}_k)|\leq\sqrt{\tr(\rho {\bf R}_j^2)}\sqrt{\tr(\rho {\bf R}_k^2)}\leq 2E_{\rm max}
\end{equation}
and
\begin{equation} \label{mean_bound}
    |\tr(\rho {\bf R}_j)|\leq\sqrt{\tr(\rho {\bf R}_j^2)}\sqrt{\tr(\rho \mathds{1}^2)}\leq \sqrt{2E_{\rm max}}
\end{equation}
for all $j,k$. Therefore, we can bound the first and second moment vectors by $\norm{{\bmu}}_\infty\leq\sqrt{2E_{\rm max}}$ and $\norm{{\bsigma}}_\infty\leq2E_{\rm max}$, respectively.

The objective of this paper is to establish a sample complexity bound for learning all parameters of an $n$-mode optical Gaussian state via a QCIS approach. This is done in two parts, each being related to one of the two steps of the QCIS scheme. The first part is to understand the optical-to-qubit transduction and show that all Gaussian parameters, the $2n^2+3n$ means and covariances of the $n$-mode state, can be estimated from qubit measurements after transduction\footnote{Throughout this paper we use the terms Gaussian parameters, means and covariances, and first and second moments interchangeably.}. The second part is to determine, given the QCIS scheme already established, how many copies of the state need to be measured to estimate all parameters to within $\epsilon$-error of the true parameter values with a total failure probability bounded by $\delta$. This objective can be stated as
\begin{equation} \label{task}
    \Pr(\norm{\widetilde{\bgamma} - \bgamma}_\infty \leq \epsilon) \geq 1-\delta,
\end{equation}
where $\bgamma$ ($\widetilde{\bgamma}$) are the true (estimated) Gaussian parameters. A similar yet stricter problem of estimating the $n$-mode density matrix for fixed trace distance while restricted to single-copy all-optical Gaussian operations and measurements has been considered in \cite{meleLearningQuantumStates2025,bittelOptimalEstimatesTrace2025,bittelEnergyindependentTomographyGaussian2025}. 

\section{Optical-to-qubit transduction}
\label{sec:transduction}
From a theoretical perspective, the key to making QCIS possible is understanding how the information describing the optical state is transferred to the qubit domain via transduction. Here we consider transducing the $j$-th mode of the optical state into a corresponding $j$-th qubit (see \cref{fig: QCIS_fig}). The transduction  for the $j$-th mode-qubit pair is modeled by a local JC interaction with zero detuning, $H_{JC,j}=g_j(a_j\sigma_j^{+}+a_j^{\dagger}\sigma_j^{-})=\bar{g}_j(Q_jX_j-P_jY_j)$ for quadrature operators $Q_j,P_j$, Pauli operators $X_j,Y_j$, and interaction strength $g_j=\sqrt{2}\bar{g}_j$ 
Therefore the total Hamiltonian acting on the initial composite state, $\rho(0)=\rho_f(0)\otimes\rho_q(0)$, is $H_{JC}=\sum_{j=1}^nH_{JC,j}$. After evolving the state for an interaction time $t$, we trace out the optical modes to find the  transduced qubit state $\rho_q(t)$. We consider an initial qubit state that is separable, $\rho_q(0)=\otimes_{j=1}^{n}\rho_q^j(0)$. This ensures that any correlation information measured in the qubits necessarily comes from the optical state through transduction. We do not explicitly consider noise in the optical-to-qubit transduction and also consider the frequency of the optical modes to be compatible with our qubits such that either no frequency conversion is necessary or that it has already been done. Some sources of noise in this transduction step could be captured by the parameter $\epsilon'$ that we introduce below, but we do not model these noise sources explicitly since the goal of this work is to establish idealized bounds on the sample complexity for Gaussian state characterization using QCIS.

\begin{figure}
    \centering
    \includegraphics[width=0.6\linewidth]{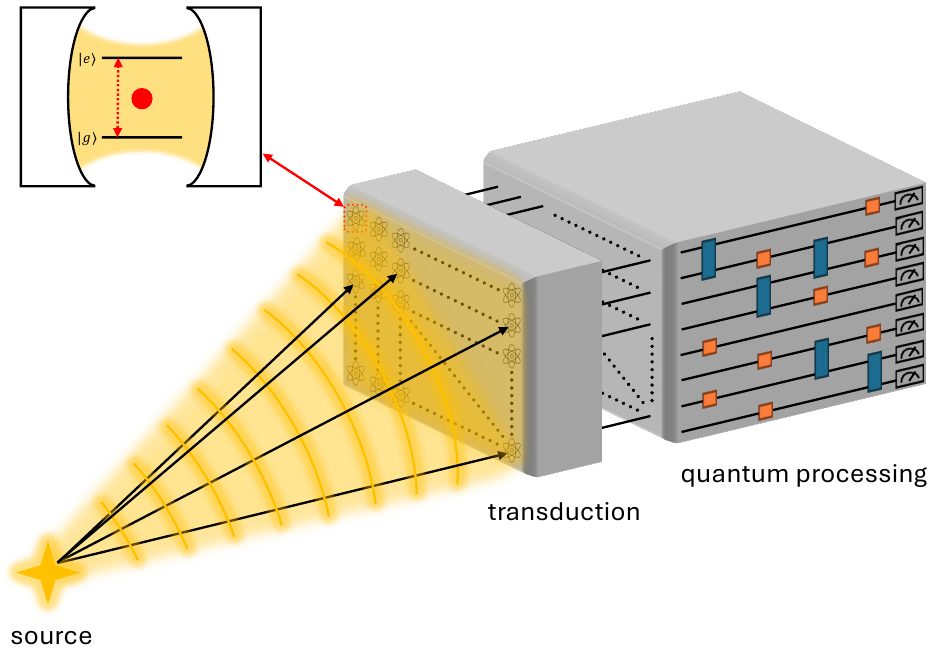}
    \caption{Schematic of the QCIS scheme to learn the quantum state of a multimode EM field by first transducing information into qubits and then performing a quantum computation on the qubits. In this example the field interacts with an array of qubits pixelated over an image plane, and thus there is a natural decompostion of the multimode field in terms of wavevector components (indicated by the black arrows). The details of the transduction will depend on  the information from the EM field to be extracted and the task at hand. In this paper we show that a simple Jaynes-Cummings interaction between a mode of the field and a qubit suffices to learn all properties of a Gaussian state of the field.
    \label{fig: QCIS_fig}}
\end{figure}

We seek an analytical solution of the JC evolution in terms of Pauli operators and Gaussian parameters so that we can extract the Gaussian parameters from Pauli measurements on $\rho_q(t)$. To do so, we use a Taylor series expansion of the unitary operator $U_{JC}(t)=e^{-itH_{JC}}$ and truncate the evolved state $\rho(t)=U_{JC}(t)\rho(0)U_{JC}^{\dagger}(t)$ after second order in $t$. This approximation yields a linear transduction mapping between one and two-qubit Pauli observables and the Gaussian parameters. Inverting this linear mapping thus provides a simple, albeit approximate, way to extract the Gaussian parameters from Pauli measurement data. Throughout, we denote the truncated state with tildes, $\widetilde{\rho}(t)$---\ie $\rho(t) = \widetilde{\rho}(t) + O(t^3)$.

In the remainder of this section, we explicitly calculate the truncated transduction evolution for $n=2$ mode Gaussian states, and show that all parameters can be estimated from sets of one- and two-qubit Pauli measurements on two distinct time-evolved initial qubit states.
In the $n$-mode case, the optical state is transduced into a register of $n$ qubits. Building off the $n=2$ results, we show that data from one- and two-qubit Pauli measurements from $\lceil{\log_2 n}\rceil+1$ distinct (product) initial states is sufficient to estimate all $n$-mode Gaussian parameters. In both cases, we assume the third-order and higher components of the evolution can be ignored entirely, and that our estimates of the Pauli observables are error-free. Both of these assumptions will be relaxed in \cref{sec:complexity} when we develop our complete QCIS procedure.

\subsection{Two-mode case} \label{sec:2-mode case}
We start by considering $n=2$, in which case we have an arbitrary two-mode Gaussian state and each of the two modes interacts independently with a qubit prepared in some initial state. In order to recover all 14 first and second moments of the two-mode state from the qubit states after interaction, we find that two different initial states of the two qubits are sufficient. We choose two distinct initial states, $\ket{\psi_1}\equiv \ket{g}_1\otimes \ket{g}_2$ and $\ket{\psi_2}\equiv\ket{+}_1\otimes \ket{+i}_2$, where $\ket{g}$ is the $-1$ eigenstate of $\sigma_z$ (\ie ground state), $\ket{+}$ is the $+1$ eigenstate of $\sigma_x$, and $\ket{+i}$ is the $+1$ eigenstate of $\sigma_y$. Note that this choice of initial states is not unique, but sufficient to estimate all Gaussian parameters. In Appendix \ref{app:2_qubit_map} we write explicit expressions for the state of the two qubits after interacting with the two modes, for each of the two initial states $\ket{\psi_1}$ and $\ket{\psi_2}$, and show that 
\begin{align}
\label{eq:o_gamma_approx}
    \bo \approx \mathcal{M} \bgamma,
\end{align}
where ${\bo}$ is a length-14 vector of one- and two-qubit Pauli expectations derived from the time-evolved qubits with one of the two possible initial states, $\bgamma$ is the length-14 vector of the two-mode Gaussian moments, and $\mathcal{M}$ is the linear mapping between the two. As explained in Appendix \ref{app:2_qubit_map}, some of the elements of $\bo$ are Pauli expectations shifted by known constants (we refer to these as \emph{shifted Pauli expectations}). This is required to turn the original affine relationship between Pauli expectations and Gaussian moments into a linear one. 

The linear approximate relationship in \cref{eq:o_gamma_approx} can be inverted to obtain estimates of the first and second moments of the two-mode Gaussian state from Pauli expectations measured on the two qubits after time evolution, $\widetilde{\pmb{\gamma}} = \mathcal{M}^{-1} {\bo}$. Note that it is sufficient to measure 14 of the 30 possible one- and two-qubit Pauli expectations on the transduced qubits to produce estimates of the 14 Gaussian moments. Performing this inversion yields the following estimates for the means:
\begin{equation} \label{two-mode solutions}
\begin{gathered}
    \widetilde{\expect{Q_1}} = \frac{\expect{Y_1}_{\ket{\psi_1}}}{2\bar{g}_1t}, \quad \widetilde{\expect{Q_2}} = \frac{\expect{Y_2}_{\ket{\psi_1}}}{2\bar{g}_2t} \\
    \widetilde{\expect{P_1}} = \frac{\expect{X_1}_{\ket{\psi_1}}}{2\bar{g}_1t}, \quad \widetilde{\expect{P_2}} = \frac{\expect{X_2}_{\ket{\psi_1}}}{2\bar{g}_2t},
\end{gathered}
\end{equation}
and for the covariance matrix elements:
\begin{equation} \label{two-mode solutions}
\begin{gathered}
    \widetilde{\expect{Q_1^2}} = \frac{1+\expect{Z_1}_{\ket{\psi_1}}}{2(\bar{g}_1t)^2}+1-\widetilde{\expect{P_1^2}} = \frac{(\expect{X_1}_{\ket{\psi_2}}+\expect{Z_1}_{\ket{\psi_1}})}{2(\bar{g}_1t)^2}+1 \\
    \widetilde{\expect{P_1^2}} = \frac{1-\expect{X_1}_{\ket{\psi_2}}}{2(\bar{g}_1t)^2}, \quad
    \widetilde{\expect{Q_2^2}} = \frac{1-\expect{Y_2}_{\ket{\psi_2}}}{2(\bar{g}_2t)^2} \\
    \widetilde{\expect{P_2^2}} = \frac{1+\expect{Z_2}_{\ket{\psi_1}}}{2(\bar{g}_2t)^2}+1-\widetilde{\expect{Q_2^2}} = \frac{(\expect{Y_2}_{\ket{\psi_2}}+\expect{Z_2}_{\ket{\psi_1}})}{2(\bar{g}_2t)^2}+1  \\
    \widetilde{\expect{\frac{Q_1P_1+P_1Q_1}{2}}} = -\frac{\expect{Y_1}_{\ket{\psi_2}}}{2(\bar{g}_1t)^2}, \quad
    \widetilde{\expect{\frac{Q_2P_2+P_2Q_2}{2}}} = -\frac{\expect{X_2}_{\ket{\psi_2}}}{2(\bar{g}_2t)^2} \\
    \widetilde{\expect{Q_1Q_2}} = \frac{\expect{Y_1Y_2}_{\ket{\psi_1}}}{4\bar{g}_1\bar{g}_2t^2}, \quad
    \widetilde{\expect{Q_1P_2}} = \frac{\expect{Y_1X_2}_{\ket{\psi_1}}}{4\bar{g}_1\bar{g}_2t^2} \\
    \widetilde{\expect{P_1Q_2}} = \frac{\expect{X_1Y_2}_{\ket{\psi_1}}}{4\bar{g}_1\bar{g}_2t^2}, \quad
    \widetilde{\expect{P_1P_2}} = \frac{\expect{X_1X_2}_{\ket{\psi_1}}}{4\bar{g}_1\bar{g}_2t^2}.
\end{gathered}
\end{equation}
In these expressions, $\expect{O}_{\ket{\psi}} \equiv \tr(O \widetilde{\rho}_q(\ket{\psi}\bra{\psi}))$ denotes the expectation of the Pauli $O$ when evolved from the initial state $\psi \in \{\psi_1,\psi_2\}$.

\subsection{$n$-mode case} \label{sec:n-mode case}
We now show how to recover all Gaussian parameters of an $n$-mode optical state by building on the two-mode case discussed above. The key observation is that the functional relationship between the Gaussian parameters defining modes $j$ and $k$ and Pauli observables of qubits $j$ and $k$ matches that derived for the two-mode case above, regardless of $j,k$. This observation is formalized in the following Lemma.

\begin{lemma}\label{nontrivial subsystem}
Consider an $n$-mode optical Gaussian state $\rho_f(0)$ interacting with an $n$-qubit state $\rho_q(0)$ via the JC interaction, $U_{JC}(t)$. The expectation value of a two-qubit Pauli observable $O_{jk}$ acting nontrivially only on qubits $j$ and $k$ is equal to the expectation value of the same observable for a two-qubit system interacting with a two-mode Gaussian state whose parameters are identical to those of modes $j,k$ in $\rho_f(0)$.
\end{lemma}
\begin{proof}
Without loss of generality, setting $j,k=1,2$ such that the qubit observable $O_{12}$ acts nontrivially only on modes 1 and 2,
\begin{align} \label{n mode evals}
    \tr(O_{12}\rho(t))
    &=
    \tr(O_{12}U_{JC}(t)[\rho_f(0)\otimes\rho_q(0)]U_{JC}^\dagger(t)) = \tr(U_{JC}^\dagger(t)O_{12}U_{JC}(t)[\rho_f(0)\otimes\rho_q(0)]) \nn \\
    &=
    \tr(U_{1,2}^\dagger(t) O_{12}U_{1,2}(t)[\rho_f(0)\otimes\rho_q(0)]) = \tr_{1,2}(U_{1,2}^\dagger(t) O_{12}U_{1,2}(t)~ \tr_{3,...,n}[\rho_f(0)\otimes\rho_q(0)]) \nn \\
    &=
    \tr_{1,2}(O_{12}U_{1,2}(t)[\rho_f^{1,2}(0)\otimes\rho_q^{1,2}(0)]U_{1,2}^\dagger(t))
\end{align}
where we have used the fact that the interaction unitary factorizes, \ie $U_{JC}(t) = \otimes_{i=1}^n e^{-it H_{JC,i}}$, and the notation $U_{1,2}(t) = \otimes_{i=1}^2 e^{-it H_{JC,i}}$. In addition, $\rho_f^{1,2}(0)$ and $\rho_q^{1,2}(0)$ are the partial trace over modes/qubits $3,...,n$ of the states $\rho_f(0)$ and $\rho_q(0)$. Crucially, the state $\rho_f^{1,2}(0)$ is again a Gaussian state, with the same means and covariances for modes 1 and 2 as prescribed by the state $\rho_f(0)$ since partial tracing over a Gaussian state yields another Gaussian state with marginalized covariance matrix and mean vector. That is, for a multimode optical Gaussian state $\rho_f^{AB}$, taking the partial trace over the $B$ modes yields a reduced state $\rho_f^{A}$ that is also Gaussian with the same means and covariances of the subsystem $A$ \cite{adessoContinuousVariableQuantum2014}. In other words, for a mean vector and covariance matrix
\begin{equation}\label{mean_vector}
{\bmu}=
\begin{pmatrix}
    {\bmu}_A \\
    {\bmu}_B
\end{pmatrix}, \quad\quad 
{\bf\Sigma}
=
\begin{pmatrix}
    {\bf\Sigma}_{A} & {\bf\Sigma}_{AB} \\
    {\bf\Sigma}_{AB}^T & {\bf\Sigma}_{B}
\end{pmatrix},
\end{equation}
tracing out subsystem $B$ gives a reduced state $\rho_f^A$ that is a Gaussian state with mean vector ${\bmu}_A$ and covariance matrix ${\bf\Sigma}_A$.
\end{proof}

Lemma \ref{nontrivial subsystem} along with the analysis in the previous subsection implies that all Gaussian moments relevant to modes $j,k$ can be recovered from Pauli expectations of qubits $j,k$ evolved under the JC interaction applied to a pair of initial states whose two-qubit reduced states on $j,k$ are $\ket{\psi_1}=\ket{g}\otimes \ket{g}$ and $\ket{\psi_2}=\ket{+}\otimes \ket{+i}$ (or $\ket{+i}\otimes \ket{+}$). In the following lemma, we establish that this can be achieved for all $n$ modes using only logarithmically many distinct states. 

\begin{lemma}\label{n-mode initial states}
There is a set of $\lceil{\log_2 n}\rceil$ $n$-qubit states such that, for every pair of qubits $j,k$, there is at least one state whose reduced state on qubits $j,k$ is either $\ket{+}\otimes \ket{+i}$ or $\ket{+i}\otimes \ket{+}$.
\end{lemma}
\begin{proof}
Let $\ell=\lceil{\log_2 n}\rceil$. We may assume that $n=2^\ell$ exactly, as in the case where $2^{\ell-1} < n \le 2^\ell$ we may use the set of states chosen for the $n = 2^\ell$ case but with qubits $n + 1, \dots 2^\ell$ traced out. 

The proof can be done via induction on $\ell$. The lemma is immediate in the $\ell=1$ ($n=2$) case. Let us now assume the lemma is true for some $\ell$. We have a set of $\ell$ $n$-qubit states such that for all $j\neq k$ with $j,k\leq n$, a pair of states reduces to either $\ket{+}_j\otimes \ket{+i}_k$ or $\ket{+i}_j\otimes \ket{+}_k$. We will construct a set of $(\ell + 1)$ $2n$-qubit states such that one of these pairings exist for every $j\neq k$ with $j,k \le 2n$.

We start by taking each state $\ket{\psi}$ from our set of $n$-qubit states and placing the state $\ket{\psi}^{\otimes 2}$ in our set of $2n$-qubit states. Now, whenever $j\neq k$ with $1 \le j,k \le n$ or $n+1 \le j,k \le 2n$, we have one of the two desired reduced states. To complete the set, we add $\ket{+}^{\otimes n} \otimes \ket{+i}^{\otimes n}$. We now have the property for all $1 \le j, k \le 2n$, and a set of size $\ell+1$.   

A visualization of this proof is provided in \cref{fig: n-mode initial states}.
\end{proof}

\begin{figure}
    \centering
    \includegraphics[width=0.5\linewidth]{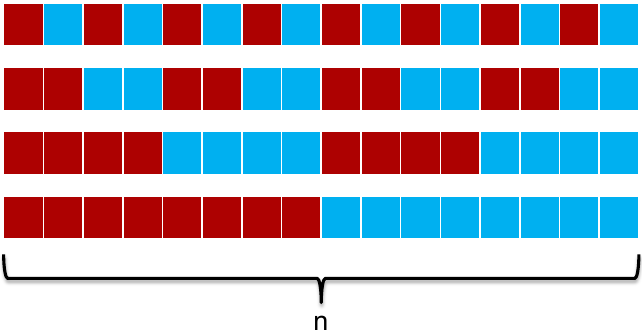}
    \caption{Visualization of the proof of Lemma \ref{n-mode initial states} for the $\ell=4$ ($9\leq n\leq16$) case. The red blocks denote initial qubit state $\ket{+}$ and the blue blocks denote initial qubit state $\ket{+i}$. There are $l=\lceil{\log_2n}\rceil$ unique initial $n$-qubit states.}
    \label{fig: n-mode initial states}
\end{figure}

We therefore conclude that \emph{all} Gaussian moments can be approximately recovered from the one- and two-qubit Pauli expectations of the transduced $n$-qubit registers of the JC interaction applied to a collection of $\lceil{\log n}\rceil+1$ states. The additional state beyond those in Lemma \ref{n-mode initial states} is $\ket{g}^{\otimes n}$, which provides the reduced state $\ket{g}\ket{g}$ for all mode pairings. Note that according to the estimation equations in \cref{sec:2-mode case}, only a subset of one- and two-qubit Pauli expectations are necessary, but for the sake of simplicity we consider measuring all one- and two-qubit Paulis as this will not change the required resources significantly. 

In the next section, we show how to use this observation to construct a sample-efficient algorithm.

\section{Iterative algorithm for Gaussian state learning}
\label{sec:complexity}

In \cref{sec:transduction} we derived a linear mapping between one- and two-qubit Pauli expectations, collected in a vector $\bo$, and Gaussian moments for the corresponding two optical modes, $\bgamma$, under $U_{JC}(t)$, by ignoring the higher-order terms of the JC interaction. In addition, we showed how this procedure can be repeated for every pair of modes in a $n$-mode Gaussian state to approximately recover all moments that define the Gaussian state. However, there were no bounds on the estimation error, due to the neglected higher-order terms of the JC interaction and statistical noise in the estimates of the Pauli observables. 

In this section we define a precise algorithm to obtain estimates of the Gaussian parameters with controlled error as a function of the resources, principally the number of copies of the Gaussian state. We focus on the problem of estimating Gaussian parameters of a two-mode state since, as discussed above, the full learning procedure is simply a repetition of this two-mode learning protocol for all pairs of modes.

The learning task we consider is to estimate every mean and covariance value to additive error $\epsilon$ (\cref{task}). The vector of (shifted) Pauli expectations, considered as a function of the Gaussian moments of the two-mode state that the qubits interact with, can be written as 
\begin{equation} \label{mapping}
{\bo}(\bgamma)=\mathcal{M}{\bgamma}+f(\bgamma),
\end{equation}
where the first term is the linear component  and $f({\bgamma})$ is the higher-order tail beyond the truncated solution of the JC interaction. 

\begin{algorithm}[t]
\caption{Iterative Gaussian moment extraction from Pauli estimates}\label{alg:alg}
\begin{algorithmic}
\State \textbf{Input:} $\widetilde{\bo}$ and $R$ \Comment{Measured (and shifted) Pauli expectations and number of estimation iterations}
\State ${\widetilde\bgamma}_0 \gets \mathcal{M}^{-1}{\widetilde \bo}$
\For{$r=1,...,R$ \do}
\State ${\widetilde\bgamma}_r \gets \mathcal{M}^{-1}({\widetilde \bo}-f({\widetilde\bgamma}_{r-1})) =  \widetilde{\bgamma}_{r-1} + \mathcal{M}^{-1}({\widetilde \bo}-\bo({\widetilde\bgamma}_{r-1}))$
\EndFor \\
\Return{${\widetilde\bgamma}={\widetilde\bgamma}_R$}
\end{algorithmic}
\end{algorithm}

The Pauli expectations can only be measured to finite precision, so suppose $\norm{{\widetilde{\bo}} - \bo}_\infty\leq\epsilon'$, where $\bo$ $(\widetilde{\bo})$ is the true (estimated) values of the Pauli expectations. Given $\widetilde\bo$, a simple estimate of the Gaussian moments, as described in \cref{sec:2-mode case}, is given by $\widetilde{\bgamma}_0 = \mathcal{M}^{-1}\widetilde{\bo}$. However, we can exploit our knowledge of the mapping in \cref{mapping} to improve the estimation error over this simple estimate by performing the iterative moment estimation algorithm described above as \cref{alg:alg}. 

In order to evaluate the estimation error, we begin with bounds on the worst-case error for the simple estimate and the iterated estimate, \ie
\begin{equation} \label{errorbound_iteration_zero}
\begin{split}
    \norm{{\widetilde\bgamma}_0-{\bgamma}}_\infty &= \norm{\mathcal{M}^{-1}({\widetilde\bo}-{\bo}(\bgamma)+f({\bgamma}))}_\infty
    \leq \norm{\mathcal{M}^{-1}}_{\infty,\infty}\Big{(}\norm{{\widetilde \bo}-{\bo}(\bgamma)}_\infty+\norm{f({\bgamma})}_\infty\Big{)},
\end{split}
\end{equation}
and
\begin{equation} \label{errorbound_iteration_r}
\begin{split}
    \norm{{\widetilde\bgamma}_r-{\bgamma}}_\infty &= \norm{\mathcal{M}^{-1}({\widetilde \bo}-f({\widetilde\bgamma}_{r-1})-{\bo}(\bgamma)+f(\bgamma))}_\infty \\
    &\leq \norm{\mathcal{M}^{-1}}_{\infty,\infty}\Big{(}\norm{{\widetilde\bo}-{ \bo}(\bgamma)}_\infty+\norm{f({\widetilde\bgamma}_{r-1})-f({\bgamma})}_\infty\Big{)}.
\end{split}
\end{equation}
Here $\norm{A}_{\infty,\infty} \equiv \max_{1\leq i \leq n}\sum_{j=1}^n |A_{ij}|$, is the induced $\infty$-norm, satisfying the property $\norm{Ab}_\infty\leq\norm{A}_{\infty,\infty}\norm{b}_\infty$ for matrix $A$ and vector $b$.
Comparing \cref{errorbound_iteration_zero} and \cref{errorbound_iteration_r}, it is evident that the accuracy of the Gaussian parameter estimates will improve as long as $\norm{f({\widetilde\bgamma}_r)-f({\bgamma})}_\infty<\norm{f({\bgamma})}_\infty$ for some chosen final $r=R$. We now introduce two lemmas that bound $\norm{\mathcal{M}^{-1}}_{\infty,\infty}$ and $\norm{f({\widetilde\bgamma}_r)-f({\bgamma})}_\infty$ so that we can explicitly evaluate the error $\norm{{\widetilde\bgamma}_r-{\bgamma}}_\infty$ and set parameters that achieve our desired accuracy.

In the following, for simplicity we assume that the JC interaction strength for each mode-qubit pair is of the same order---\ie $\bar{g}_i=\Theta(g)$ for all $i$. The following Lemma bounds the norms of submatrices within $\mathcal{M}^{-1}$, which is useful for separating the scaling of errors for estimates of first- and second-order Gaussian moments.
\begin{lemma}\label{lemma1}
    $\mathcal{M}^{-1}$ decomposes as $\left(\begin{array}{c}
         \mathcal{A} \\
         \mathcal{B} 
    \end{array}\right)$ where $\dim(\mathcal{A})=4\times 14$ and $\dim(\mathcal{B})=10\times 14$, and
    \begin{equation}
    \begin{split}
    &\norm{\mathcal{A}}_{\infty,\infty}=\mathcal{O}\left(\frac{1}{gt}\right), \\
    &\norm{\mathcal{B}}_{\infty,\infty}=\mathcal{O}\left(\frac{1}{(gt)^2}\right).
    \end{split}
    \end{equation}
\end{lemma}
\begin{proof}
This is proved by direct computation. In Appendix \ref{app:2_qubit_map} we display $\mathcal{M}^{-1}$ to show the above norm scalings.
\end{proof}

Now that we have bounded the inversion on the linear part of our solution that gives us our estimates of $\bgamma$, we have to also bound how much an error on these estimates perturbs the nonlinear part $f({\bgamma})$. Recall that we follow the convention of writing $\bgamma = (\bmu,\bsigma)^T$, where $\bmu$ is the vector of means and $\bsigma$ the vector of (co)variances. We will likewise write ${\widetilde\bgamma}=({\widetilde\bmu},{\widetilde\bsigma})^T$ for an estimate of these parameters.

\begin{lemma}\label{lemma2}
    There is a constant $C>0$ such that if $gt\leq \frac{1}{C\sqrt{E_{\rm max}}}$, then for all estimates ${\widetilde\bgamma}=({\widetilde\bmu},{\widetilde\bsigma})^T$ with $\norm{\widetilde\bmu}_{\infty}<2\sqrt{2E_{\rm max}}$ and $\norm{\widetilde\bsigma}_{\infty}<4E_{\rm max}$,
    \begin{equation}
    \begin{split}
    \norm{f({\widetilde\bgamma})-f({\bgamma})}_\infty=\mathcal{O}\Big{(}(gt)^3\Big{(}E_{\rm max}\norm{{\widetilde\bmu}-{\bmu}}_\infty 
    +\sqrt{E_{\rm max}}\norm{{\widetilde\bsigma}-{\bsigma}}_\infty \Big{)}\Big{)}.
    \end{split}
    \end{equation}
\end{lemma}
\begin{proof}
In order to bound the higher-order contributions collected in $f(\bgamma)$ let us begin by studying the exact form of the Pauli observables that form of the entries of $\bo(\gamma)$. The Pauli observables are under a time evolved state $\rho(t)$, which we can write using a Taylor expansion of the unitary time evolution (or equivalently Campbell's identity \cite{Achilles2012}) as
\begin{align}
\begin{split}
    \rho(t)=e^{-iH_{JC}t}\rho(0)e^{iH_{JC}t} &=\sum_{k=0}^\infty\frac{(-it)^k}{k!}[H_{JC},\rho(0)]_k \\
\end{split}
\end{align}
where $H_{JC}=\sum_{l=1}^nH_{JC,l}=\sum_{l=1}^n\bar{g}_l(Q_lX_l-P_lY_l)$ and $[H,\rho(0)]_k \equiv [H,[H,...[H,\rho(0)]...]]$ is the $k$-nested commutator. The initial state $\rho(0)$ is the $n$-mode Gaussian state tensored with a product initial state of the $n$ qubits---\ie $\rho(0) = \rho_f(0) \otimes_{l=1}^n \rho_q^l(0)$. Taking the expectation of the evolved state with a qubit observable $O_m$ acting on a set of qubits $m\subseteq\{1,...n\}$:
\begin{align}\label{Pauli_expectation}
\begin{split}
    \tr(O_m\rho(t)) &= \sum_{k=0}^\infty\frac{(-it)^k}{k!} \tr(O_m[H_{JC},\rho(0)]_k) = \sum_{k=0}^\infty\frac{(-it)^k}{k!} \tr(O_m[\bar{H},\bar{\rho}(0)]_k), 
\end{split}
\end{align}
where $\bar{H}=\sum_{l\in m}H_l$ and $\bar\rho(0) = \tr_{\bar{m}}(\rho(0)) = \rho_f^m(0)\otimes_{l\in m}\rho_q^l(0)$ is the reduced initial density matrix obtained after tracing out all modes and qubits not in the set $m$. In the second equality in \cref{Pauli_expectation} we have used 
$\tr(O_m[H_{JC},\rho(0)]_k)=\tr(O_m[\bar{H},\rho(0)]_k)$ since $\tr(O_m[H_{JC},\rho(0)]_k)=(-1)^k \tr(\rho(0)[H_{JC},O_m]_k)$ and $[H_l,O_m]=0$ if $l\not\in m$. 
Importantly, the reduced density matrix of the field $\rho_f^m$ is a Gaussian state on $m$ modes, as argued in \cref{sec:n-mode case}.

We next use the property that $[\bar{H},\bar{\rho}(0)]_k=\sum_{j=0}^k\binom{k}{j}(-1)^j\bar{H}^{k-j}\bar{\rho}(0)\bar{H}^j$ as well as the multinomial theorem on $\bar{H}^{k-j}$ and $\bar{H}^j$ to yield
\begin{align}
\begin{split}
    \tr(O_m\rho(t)) &= \sum_{k=0}^\infty\frac{(-it)^k}{k!}\sum_{j=0}^k\binom{k}{j}(-1)^j \\ & \ \ \ \ \ \ \tr\Big{(}O_m\sum_{\sum a_l=k-j}\binom{k-j}{a_1,...,a_m}\Big{(}\prod_{l\in m}H_l^{a_l}\Big{)}\bar{\rho}(0)\sum_{\sum b_i=j}\binom{j}{b_1,...,b_m}\Big{(}\prod_{i\in m}H_i^{b_i}\Big{)}\Big{)} \\
    &= \sum_{k=0}^\infty\sum_{j=0}^k\sum_{\sum a_l=k-j}\sum_{\sum b_i=j}\frac{(-it)^k}{k!}(-1)^j\binom{k}{j}\binom{k-j}{a_1,...,a_m}\binom{j}{b_1,...,b_m} \\ & \ \ \ \ \ \ \tr\Big{(}O_m\Big{(}\prod_{l\in m}H_l^{a_l}\Big{)}\bar{\rho}(0)\Big{(}\prod_{i\in m}H_i^{b_i}\Big{)}\Big{)}. \\
\end{split}
\end{align}
Applying the binomial theorem to $H_l^{a_l}$ gives
\begin{align}
\begin{split}
    H_l^{a_l} &= \bar{g}_l^{a_l}\sum_{\mu_l=0}^{a_l}\binom{a_l}{\mu_l}(Q_lX_l)^{a_l-\mu_l}(-P_lY_l)^{\mu_l} = \bar{g}_l^{a_l}\sum_{\mu_l=0}^{a_l}\binom{a_l}{\mu_l}Q_l^{a_l-\mu_l}P_l^{\mu_l}R_l \\
\end{split}
\end{align}
where $R_l\in\{\pm1,\pm i\}\times\{\mathds{1},X_l,Y_l,Z_l\}$. This and the likewise expansion for $H_i^{b_i}$ allow us to write the expectation of $O_m$ as
\begin{align}
\begin{split}
    \tr(O_m\rho(t)) &= \sum_{k=0}^\infty\sum_{j=0}^k\sum_{\sum a_l=k-j}\sum_{\sum b_i=j}\sum_{\mu_l=0}^{a_l}\sum_{\beta_i=0}^{b_i}(-it)^k(-1)^j \\ & \ \ \ \ \ \ \tr\Big{(}O_m\Big{(}\prod_{l\in m}\frac{\bar{g}_l^{a_l}}{(a_l-\mu_l)!\mu_l!}Q_l^{a_l-\mu_l}P_l^{\mu_l}R_l\Big{)}\bar{\rho}(0)\Big{(}\prod_{i\in m}\frac{\bar{g}_i^{b_i}}{(b_i-\beta_i)!\beta_i!}Q_i^{b_i-\beta_i}P_i^{\beta_i}R_i\Big{)}\Big{)} \\
\end{split}
\end{align}
where after simplification most factorials cancel out. Evaluating the trace,
\begin{align}\label{Pauli_expectation_final}
\begin{split}
    \tr(O_m\rho(t)) &= \sum_{k=0}^\infty\sum_{j=0}^k\sum_{\sum a_l=k-j}\sum_{\sum b_i=j}\sum_{\mu_l=0}^{a_l}\sum_{\beta_i=0}^{b_i}(-it)^k(-1)^j \\ & \ \ \ \ \ \ \Big\langle\prod_{i\in m}\prod_{l\in m}\frac{\bar{g}_l^{a_l}\bar{g}_i^{b_i}}{(a_l-\mu_l)!\mu_l!(b_i-\beta_i)!\beta_i!}Q_i^{b_i-\beta_i}P_i^{\beta_i}Q_l^{a_l-\mu_l}P_l^{\mu_l}\Big\rangle\Big\langle R_iO_mR_l\Big\rangle
\end{split}
\end{align}
where the expectations are over the initial $|m|$-mode and $|m|$-qubit states, $\rho_f^m(0)$ and $\otimes_{k\in m}\rho_q^m(0)$, respectively. Note that $\expect{R_iO_mR_l}$ is $\mathcal{O}(1)$ in cases of interest to us where $|m|\leq 2$. 

$\expect{\prod_{i\in m}\prod_{l\in m}Q_i^{b_i-\beta_i}P_i^{\beta_i}Q_l^{a_l-\mu_l}P_l^{\mu_l}}$ is an order $k$ moment. Because the state $\rho_f(0)$ is Gaussian, we can use Isserlis' theorem to expand any order $k\geq3$ moment in terms of first- and second-order moments $\bmu$ and $\bsigma$. However since Isserlis' theorem assumes zero-mean Gaussians, we first need to rewrite the raw moments in terms of centralized moments. This leads to an expansion of the order $k$ moment into a sum of $2^{k-1}$ terms of products of (raw) first moments and centralized even moments (there are $2^k$ terms total but all odd central moments are zero). By applying Isserlis' theorem, a centralized even moment of order $k$ is written as a sum of $(k-1)!!$ products of $k/2$ second-order centralized moments \cite{Isserlis1918}.

Putting this together, $\expect{\prod_{i\in m}\prod_{l\in m}Q_i^{b_i-\beta_i}P_i^{\beta_i}Q_l^{a_l-\mu_l}P_l^{\mu_l}}$ in \cref{Pauli_expectation_final} can be written as a sum of no more than $2^{k-1}(k-1)!!$ terms of products of first and second moments contained in $\bmu$ and $\bsigma$ (degrees adding up to $k$). Each of these terms is multiplied by a coefficient of at most $\frac{(gt)^k}{(\lfloor k/{(4|m|)}\rfloor!)^{4|m|}}$. Now consider replacing one first- or second-order moment in one of these terms by estimates that are wrong by $\norm{{\widetilde\bmu}-{\bmu}}_{\infty}$ and $\norm{{\widetilde\bsigma}-{\bsigma}}_{\infty}$. This perturbs an order $k$ term by at most $\mathcal{O}\Big{(}2^{k-1}(2E_{\rm max})^{(k-1)/2}\norm{{\widetilde\bmu}-{\bmu}}_{\infty}+2^{k-2}(2E_{\rm max})^{(k-2)/2}\norm{{\widetilde\bsigma}-{\bsigma}}_{\infty}\Big{)}$, since we can upper bound both the first moments and their estimates by $2\sqrt{2E_{\rm max}}$ and the second moments and their estimates by $4E_{\rm max}$. Summing over all these terms, we can bound the difference of the higher order terms in the expression for a Pauli expectation in terms of the differences in Gaussian moments as
\begin{equation}
\begin{split}
    \norm{f({\widetilde\bgamma})-f({\bgamma})}_\infty &\leq \sum_{k=3}^{\infty}\frac{\mathcal{O}(gt)^kk!!}{(\lfloor k/{4|m|}\rfloor!)^{4|m|}} \Big{(}E_{\rm max}^{(k-1)/2}\norm{{\widetilde\bmu}-{\bmu}}_{\infty}+E_{\rm max}^{(k-2)/2}\norm{{\widetilde\bsigma}-{\bsigma}}_{\infty}\Big{)} \\
    &=\mathcal{O}\Big{(}(gt)^3\Big{(}E_{\rm max}\norm{{\widetilde\bmu}-{\bmu}}_\infty +\sqrt{E_{\rm max}}\norm{{\widetilde\bsigma}-{\bsigma}}_\infty \Big{)}\Big{)}
\end{split}
\end{equation}
if we set $gt\leq\frac{1}{C\sqrt{E_{\rm max}}}$ for sufficiently large $C$.
\end{proof}

Given Lemmas \ref{lemma1} and \ref{lemma2}, we can now evaluate the error bound in \cref{errorbound_iteration_r}. 

\begin{theorem}\label{epsilon_scaling}
    There are choices of $gt=\Theta(\frac{1}{\sqrt{E_{\rm max}}})$ and $\epsilon'=\Theta(\frac{\epsilon}{E_{\rm max}})$ such that running \cref{alg:alg} for $R=\log{\frac{E_{\rm max}}{\epsilon}}$ returns $\widetilde\bgamma$ with $\norm{{\widetilde\bmu}-{\bmu}}_\infty\leq \frac{\epsilon}{\sqrt{E_{\rm max}}}$ and $\norm{{\widetilde\bsigma}-{\bsigma}}_\infty\leq\epsilon$.
\end{theorem}
\begin{proof}
The bounds in Lemma \ref{lemma1} motivate us to split analysis of the error in estimating $\bgamma$ into two parts corresponding to the means $\bmu$ and covariances $\bsigma$ since their estimation errors scale differently. Similar to \cref{errorbound_iteration_r}, the error bounds for these two components are:
\begin{align}
    \norm{\widetilde{\bmu}_r - \bmu}_\infty &\leq \norm{\mathcal{A}}_{\infty,\infty}\left(\norm{\widetilde{\bo}-\bo}_{\infty} + \norm{f_\bmu(\widetilde{\bgamma}_{r-1}) - f_\bmu(\bgamma)}_{\infty}\right) \nn \\
    &\leq \norm{\mathcal{A}}_{\infty,\infty}\left(\norm{\widetilde{\bo}-\bo}_{\infty} + \norm{f(\widetilde{\bgamma}_{r-1}) - f(\bgamma)}_{\infty}\right), \nn \\
    \norm{\widetilde{\bsigma}_r - \bsigma}_\infty &\leq \norm{\mathcal{B}}_{\infty,\infty}\left(\norm{\widetilde{\bo}-\bo}_{\infty} + \norm{f_\bsigma(\widetilde{\bgamma}_{r-1}) - f_\bsigma(\bgamma)}_{\infty}\right) \nn \\
    &\leq \norm{\mathcal{B}}_{\infty,\infty}\left(\norm{\widetilde{\bo}-\bo}_{\infty} + \norm{f(\widetilde{\bgamma}_{r-1}) - f(\bgamma)}_{\infty}\right),
\end{align}
where $f_{\bmu} (f_\bsigma)$ are the first four (last ten) elements of the vector $f$ that contains the non-linear correction terms. The above expressions are also valid for $r=0$ if we set $\widetilde{\bgamma}_{-1}=0$ since $f(0)=0$.  
Now, we can choose $gt=\Theta(\frac{1}{\sqrt{E_{\rm max}}})$ and $\epsilon'=\Theta(\frac{\epsilon}{E_{\rm max}})$ such that, by applying Lemmas \ref{lemma1} and \ref{lemma2}, we have
\begin{align}
    \norm{\widetilde{\bmu}_r - \bmu}_\infty &\leq \frac{\epsilon}{4\sqrt{E_{\rm max}}} + \frac{1}{4}\left(\norm{\bmu_{r-1} - \bmu}_{\infty} + \frac{1}{\sqrt{E_{\rm max}}}\norm{\widetilde{\bsigma}_{r-1} - \bsigma}_\infty\right), \nn \\
    \norm{\widetilde{\bsigma}_r - \bsigma}_\infty &\leq \frac{\epsilon}{4} + \frac{1}{4}\left( \sqrt{E_{\rm max}}\norm{\widetilde{\bmu}_{r-1} - \bmu}_\infty + \norm{\widetilde{\bsigma}_{r-1} - \bsigma}_\infty\right)
\end{align}

Next we prove by induction that $\norm{{\widetilde\bmu}_r-{\bmu}}_\infty \leq \epsilon/(2\sqrt{E_{\rm max}})+2^{-r}\sqrt{E_{\rm max}}$ and $\norm{{\widetilde\bsigma}_r-{\bsigma}}_\infty\leq\epsilon/2+2^{-r}E_{\rm max}$ for all integers $r\geq0$.

First we evaluate for $r=0$. Recall that $\bgamma_{-1}=0$, and using the bounds on means and covariance matrix elements, \cref{covariance_bound,mean_bound}, we get:
\begin{align}
    \norm{\widetilde{\bmu}_0 - \bmu}_\infty &\leq \frac{\epsilon}{4\sqrt{E_{\rm max}}} + \left(\frac{2+\sqrt{2}}{4}\right)\sqrt{E_{\rm max}} \nn \\
    \norm{\widetilde{\bsigma}_0 - \bsigma}_\infty &\leq \frac{\epsilon}{4} + \left(\frac{2+\sqrt{2}}{4}\right)E_{\rm max}
\end{align}
Now assume the bounds hold for $r$, and consider the $r+1$ case:
\begin{align}
    \norm{\widetilde{\bmu}_{r+1} - \bmu}_\infty &\leq \frac{\epsilon}{4\sqrt{E_{\rm max}}} + \frac{1}{4}\left(\norm{\bmu_{r} - \bmu}_{\infty} + \frac{1}{\sqrt{E_{\rm max}}}\norm{\widetilde{\bsigma}_{r} - \bsigma}_\infty\right) \nn \\
    &\leq\frac{\epsilon}{2\sqrt{E_{\rm max}}}+2^{-(r+1)}\sqrt{E_{\rm max}},
\end{align}
where we have used the assumption to obtain the second inequality. Similarly, 
\begin{align}
    \norm{\widetilde{\bsigma}_{r+1} - \bsigma}_\infty &\leq \frac{\epsilon}{4} + \frac{1}{4}\left( \sqrt{E_{\rm max}}\norm{\widetilde{\bmu}_{r} - \bmu}_\infty + \norm{\widetilde{\bsigma}_{r} - \bsigma}_\infty\right) \nn\\
    &\leq \frac{\epsilon}{2} + 2^{-(r+1)}E_{\rm max},
\end{align}
as needed. Given these bounds on the iterative estimates, setting $R=\log{(E_{\rm max}/\epsilon)}$ completes the proof of \cref{epsilon_scaling}.
\end{proof}

To demonstrate the validity of \cref{epsilon_scaling}, we simulate the transduction of an arbitrary two-mode Gaussian state and run \cref{alg:alg} on the simulated Pauli expectation measurement data to show how the error quickly diminishes with subsequent iterations. A general Gaussian state can be formed by squeezing then displacing a thermal state via $\rho_f=DS\rho_{\rm th}S^{\dagger}D^{\dagger}$, where $\rho_{\rm th}$ is a thermal state, $S$ is a squeezing operator, and $D$ is a displacement operator. For the two-mode case: the thermal state is $\rho_{\rm th}=\rho_{{\rm th},1}(\bar{n}_1)\otimes\rho_{{\rm th},2}(\bar{n}_2)$ where $\rho_{{\rm th},j}(\bar{n}_j)=\sum_{n_j=0}^{\infty}P_{n_j} \ketbra{n_j}{n_j}$ with $P_{n_j}=\frac{\bar{n}_j^{n_j}}{(\bar{n}_j+1)^{n_j+1}}$ and thermal occupation number $\bar{n}_j$, the squeezing operator is $S=S(z_{12})S(z_{22})S(z_{11})$ where $S(z_{jk})=e^{(z_{jk}^*a_ja_k-z_{jk}a_j^{\dagger}a_k^{\dagger})/2}$, and the displacement operator is $D=D(\alpha_2)D(\alpha_1)$ where $D(\alpha_j)=e^{\alpha_ja_j^{\dagger} - \alpha_j^*a_j}$. We choose arbitrary non-zero values for all seven parameters, $\bar{n}_1, \bar{n}_2, z_{11}, z_{22}, z_{12}, \alpha_1$, and $\alpha_2$, to generate an arbitrary two-mode Gaussian state $\rho_f$ in the Fock basis up to some Fock number truncation. We then simulate the optical-to-qubit transduction of $\rho_f$ via the JC interaction with two distinct initial states of two-qubits, as outlined in \cref{sec:2-mode case}. For this we choose $gt=0.01$. This enforces the bound on $gt$ from Lemma \ref{lemma2} such that \cref{epsilon_scaling} is valid for any $E_{\rm max}$ on the order of $\lesssim1000$'s per mode. The simulated Pauli measurement data $\widetilde{\bo}$ that this generates is then fed into \cref{alg:alg} to produce Gaussian moment estimates. In \cref{fig:simulation}, we plot the maximum Gaussian moment estimation error, $\norm{\widetilde{\bgamma} - \bgamma}_\infty$, as a function of algorithm iteration, $r$, showing that the algorithm converges rapidly for general two-mode Gaussian states.
\begin{figure}
    \centering
    \includegraphics[width=0.5\linewidth]{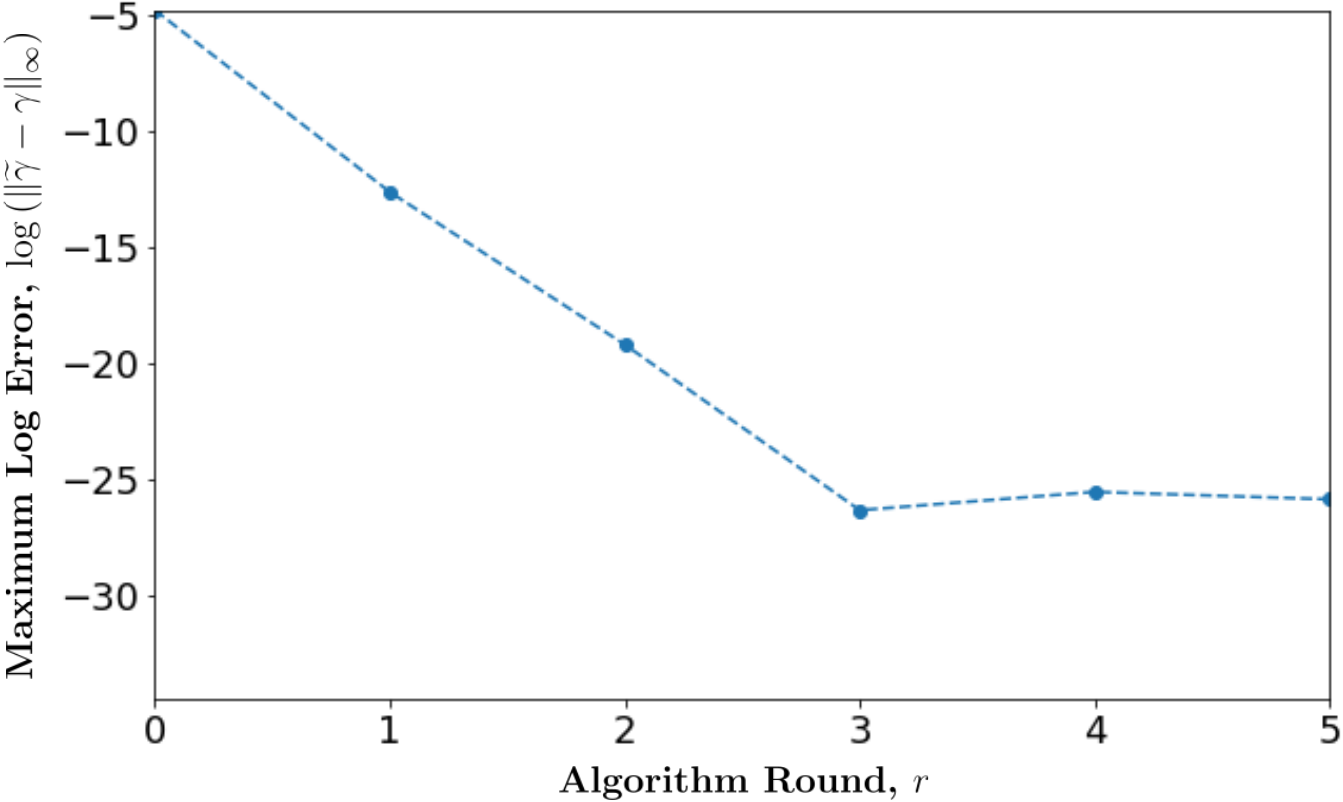}
    \caption{Plot of the maximum Gaussian parameter estimation error, $\norm{\widetilde{\bgamma} - \bgamma}_\infty$, versus the algorithm round, $r$, for a two-mode Gaussian state with unitary parameter values of $\bar{n}_1=0.3, \bar{n}_2=0.5, z_{11}=0.2, z_{22}=0.3, z_{12}=0.4, \alpha_1=\frac{3(1-4i)}{10\sqrt{2}}$, and $\alpha_2=\frac{1(3+2i)}{5\sqrt{2}}$. A Fock number truncation of 60 is chosen for each mode and $gt=0.01$ for the transduction.}
    \label{fig:simulation}
\end{figure}

We are now ready to prove the main result of this paper, \cref{main_result}, restated here for convenience.

\begin{manualtheorem}{1}
With probability $1-\delta$, optical-to-qubit transduction followed by one- and two-qubit Pauli measurements can estimate each of the $M=2n^2+3n$ means and covariances of an $n$-mode optical state to additive error $\epsilon$ using 
\begin{align}
\label{eq:T_final}
    T=\mathcal{O}\Big{(}\frac{E_{\rm max}^2 \log(n/\delta)\lceil\log n \rceil}{\epsilon^2}\Big{)}
\end{align} 
copies, where $E_{\rm max}$ is the maximum energy per mode.
\end{manualtheorem}

\begin{proof}
The $M$ parameters are estimated by combining the estimates of all pairwise moments. According to \cref{epsilon_scaling} all pairwise moments can be estimated with maximum error $\epsilon$ by choosing JC interaction time $gt=\Theta(\frac{1}{\sqrt{E_{\rm max}}})$, running the iterative estimation algorithm with $R=\log \frac{E_{\rm max}}{\epsilon}$, and setting $\epsilon' = \norm{\widetilde{\bo} - \bo}_\infty = \Theta(\frac{\epsilon}{E_{\rm max}})$. This last setting bounds the error on the one- and two-qubit Pauli measurements performed on the $n$-qubit register after transduction. We perform these measurements using the classical shadows protocol, which guarantees with probability $1-\delta$ that $B$ $k$-local Pauli observables can be measured to additive error $\epsilon'$ using $\Theta(3^k\log{(B/\delta)}/\epsilon'^2)$ copies of a state \cite{huangPredictingManyProperties2020, Huang2021-xl}. In our case $k=2$ and $B=9 {n \choose 2} + 3n = \frac{9n(n-1)}{2} + 3n$ since this is the number of one- and two-qubit observables for an $n$-qubit register. Using these parameters and the choice of $\epsilon'$ needed from \cref{epsilon_scaling} yields $T = \mathcal{O}\left(\frac{E_{\max}^2 \log(n/\delta)}{\epsilon^2}\right)$ copies of the state. The remaining factor in \cref{eq:T_final} is from the fact that according to Lemma \ref{n-mode initial states} $\lceil \log_2 n \rceil+1$ initial states of the qubit register are required to ensure that there are sufficient Pauli observables to form an invertible linear system that recovers the Gaussian moments for every pair of qubits.  
\end{proof}

We conclude this section by formalizing the complete QCIS protocol as developed in \cref{sec:transduction,sec:complexity} as \cref{alg2}.

\begin{center}
    \begin{algorithm}
        \begin{algorithmic}[1]
            \Require $T$ copies of an $n$-mode Gaussian state $\rho_f$ with single-mode energy bounded by $E_{\rm max}$.
            \State Let $T'=\frac{T}{\lceil \log_2 n \rceil+1}$.
            \State Initialize the $n$-qubit state to $\rho_q(0)=\ketbra{g}{g}^{\otimes n}$.
             \For{$r=1,...,T'$}
                \State Transduce a copy of $\rho_f$ into the $n$-qubit register via the Jaynes-Cummings interaction with an interaction time $gt=\Theta(\frac{1}{\sqrt{E_{\rm max}}})$ as defined in \cref{epsilon_scaling}.
                \State Perform measurements on the $n$-qubit transduced state $\rho_q(t)$ as prescribed by the classical shadows protocol.
             \EndFor
             \State Construct the classical shadow of size $T'$ from the measurement data in step 5 and estimate the one- and two-qubit Pauli observables that form the vector $\widetilde\bo$ for every pair of qubits. The observables required are listed in \cref{Pauli_ordering}.
             \For{each of the $\lceil \log_2 n \rceil$ states described in Lemma \ref{n-mode initial states} (\eg $\rho_q(0)=\ketbra{+}{+}^{\otimes n/2}\ketbra{+i}{+i}^{\otimes n/2}$)}
              \State Prepare the initial $n$-qubit state.
              \State Repeat steps $3-7$.
               \EndFor
               \State For each pair of qubits, construct $\widetilde\bo$ from the appropriate Pauli observable estimates and execute \cref{alg:alg}.
               \State Combine Gaussian parameter estimates from each pair, averaging over repeated estimates, and output all $2n^2+3n$ Gaussian moment estimates of $\rho_f$.
        \end{algorithmic}
        \caption{The QCIS protocol for estimating all parameters of an $n$-mode Gaussian state. \label{alg2}}
    \end{algorithm}
\end{center}

\section{Sample complexity of Gaussian state learning via CV classical shadows}
\label{sec:CV shadows}

We now restate the sample complexity bounds for CV classical shadows from \cite{gandhariPrecisionBoundsContinuousVariable2024,beckerClassicalShadowTomography2024} in terms of our task, \cref{task}, in order to compare directly with our main result, \cref{main_result}. These all-optical strategies use randomized Gaussian operations and single-copy local measurements such as homodyne or photon-number resolving detection.

\begin{theorem}\label{CV shadows theorem}
[Via Corollary 2 of Gandhari \etal\ \cite{gandhariPrecisionBoundsContinuousVariable2024}] CV shadow tomography can estimate the $M=2n^2+3n$ means and covariances of an n-mode optical state to additive error $\epsilon$ with probability $1-\delta$ using $T=\mathcal{O}\Big{(}\frac{\max_{i=1}^M||O_i||_{\infty}^2 E_{\rm max}^{8}\log(n/\delta)}{\epsilon^2}\Big{)}$ copies of the state, where $O_i$ are the observables corresponding to the Gaussian moments.
\end{theorem}
\begin{proof}
Note the similar form of the bound to \cref{main_result}, which is not surprising considering that the sample complexity in this result and our result ultimately stem from the resources required for classical shadow reconstruction of observables. From \emph{Corollary 2} (Eq. (34)) of \cite{gandhariPrecisionBoundsContinuousVariable2024}, the sample complexity to estimate $M$ multimode observables ${O_1,...,O_M}$ local to ${k_1,...,k_M}$ modes to precision $\epsilon$ with probability of failure bounded by $\delta$ is $T=\mathcal{O}(\max_i||O_i||_{\infty}^2N^{2k_i}\nu_1^{2k_i}(k_i\log(2N)+\log(M/\delta))/\epsilon^2)$.

For Gaussian states, $k_i\leq2$ and $N$ is the single-mode Fock truncation number which can be chosen as $\Theta(E_{\rm max})$ to scale with the maximum single-mode energy. $\nu_1^2$ is the single-mode variance of the estimator, which is also bounded by $\mathcal{O}(E_{\rm max}^2)$. This can be seen by expanding the variance of a two-mode observable via Isserlis' theorem and using the bounds on first and second moments, \eqref{covariance_bound} and \eqref{mean_bound}. Putting this all together, we find that $T=\mathcal{O}(\max_i||O_i||_{\infty}^2 E_{\rm max}^{8}(\log(2E_{\rm max})+\log(n/\delta))/\epsilon^2)$. The term $\max_i||O_i||_{\infty}^2$ is non-trivial to evaluate explicitly but could be approximated (\eg the energy cutoff could be translated into a Fock space upper bound and the observables could be treated in a finite-dimensional Hilbert space). We do not attempt to do this here but simply note that this term cannot have an $n$-dependence because the $M$ observables act on at most two modes. 
\end{proof}

Comparing this result to \cref{main_result}, we see that the QCIS upper bound has the same dependence on $n$ but has a polynomially better dependence on $E_{\rm max}$.

\section{Conclusion} \label{sec:conclusion}
We have established sample complexity bounds on learning all parameters of $n$-mode Gaussian optical states within a QCIS framework whereby the optical fields are transduced to qubits and a quantum-computer-enabled measurement is performed on the qubits. We showed that by using a JC interaction between each field mode and an independent qubit and then performing one- and two-qubit Pauli measurements on the $n$-qubit register using classical shadows, one can recover all Gaussian moments in the original $n$-mode field. We rigorously proved that the worst-case estimation error can be bounded by $\epsilon$ by controlling the JC interaction duration and using an iterative estimation algorithm with a number of samples of the Gaussian state that scales logarithmically in $n$. While this scaling with $n$ is matched by recent continuous-variable classical shadows protocols that only perform optical operations and measurements \cite{gandhariPrecisionBoundsContinuousVariable2024}, our protocol provides a polynomial improvement over these in terms of the sample complexity dependence on $E_{\rm max}$, a bound on the energy per mode. 

The overwhelming majority of low-intensity imaging and sensing applications use Gaussian states of the electromagnetic field. The existence of a CV shadows protocol, whose sample complexity scales as $\mathcal{O}(\log n)$, precludes exponential learning advantages in terms of $n$ for Gaussian states using QCIS. While our work shows that one still can have an advantage in terms of $E_{\rm max}$ using a QCIS scheme, relative to the best-known optical measurement scheme, it is unlikely to be sufficient to motivate the added hardware complexity of QCIS. This rules out a large set of applications for which one might consider QCIS in order to obtain an exponential advantage, and instead shifts focus to more advanced imaging and sensing applications that involve non-Gaussian states of the electromagnetic field. 

There are several directions for extending the work in this paper. First, we did not consider experimental imperfections beyond the statistical error associated to Pauli expectation measurements using classical shadows, since we were interested in establishing fundamental bounds. Experimental imperfections in transduction could be modeled, as was done in \cite{crossmanQuantumComputerenabledReceivers2024}, in future work. Second, we only considered the JC interaction since it is the simplest and most commonly encountered atom-field interaction. Other strategies that transduce a mode into multiple qubits (\eg using the Tavis-Cummings interaction) could provide benefits. While we do not believe the $\mathcal{O}(\log n)$ scaling of sample complexity could be improved, it is possible that the estimation scheme could be simplified with other interactions. 

\begin{acknowledgements}
    This work was supported by the U.S. Department of Energy, Office of Science, Office of Advanced Scientific Computing Research through the EXPRESS Program and the Accelerated Research in Quantum Computing Program MACH-Q Project. Sandia National Laboratories is a multimission laboratory managed and operated by National Technology and Engineering Solutions of Sandia LLC, a wholly owned subsidiary of Honeywell International Inc. for the U.S. Department of Energy’s National Nuclear Security Administration under contract DE-NA0003525. This paper describes objective technical results and analysis. Any subjective views or opinions that might be expressed in the paper do not necessarily represent the views of the U.S. Department of Energy or the United States Government.
\end{acknowledgements}

\bibliography{biblio}

\appendix
\section{Expression for two-qubit transduced state and Gaussian moment estimation}
\label{app:2_qubit_map}
In this Appendix, we explicitly calculate how the $2$-mode, $2$-qubit system evolves under a second-order approximation to the JC interaction, and how the Gaussian moments of the modes can be estimated by inverting the relationship between them and two-qubit Pauli expectations. The evolved composite state is solved via
\begin{align}
\rho(t) &= U_{JC}(t) \rho(0) U\dg_{JC}(t) =\sum_{k=0}^\infty\frac{(-it)^k}{k!}[H_{JC},\rho(0)]_k \\
&= \underbrace{\rho(0) -it[H_{JC}, \rho(0)] - \frac{t^2}{2}[H_{JC},[H_{JC}, \rho(0)]]}_{\widetilde{\rho}(t)} + O(t^3),
\end{align}
where for this two-mode case, $U_{JC}(t)=e^{-itH_{JC}}$, $H_{JC} = \sum_{j=1}^2 \bar{g}_j(Q_jX_j-P_jY_j)$, and $\rho(0) = \rho_f \otimes_{j=1}^2 \rho_q^j(0)$. We calculate the first two nontrivial terms in this expansion for an arbitrary two-mode optical state and two initial states of the qubits: $\ket{\psi_1}\equiv \ket{g}_1\otimes \ket{g}_2$ and $\ket{\psi_2}\equiv\ket{+}_1\otimes \ket{+i}_2$. In each case, the (approximate) time-evolved reduced state of the two-qubits, expressed in terms of a Pauli expansion, takes the form:
\begin{align} \label{2-mode state_gg}
\begin{split}
    \widetilde{\rho}_q\Big{(}\ketbra{g}{g}_1\otimes\ketbra{g}{g}_2\Big{)}
    = \frac{1}{4}\bigg{(}\mathds{1} &+
    2\bar{g}_1t\expect{P_1}X_1 +
    2\bar{g}_1t\expect{Q_1}Y_1 - \Big{(}1-2(\bar{g}_1t)^2\big{(}\expect{Q_1^2}+\expect{P_1^2}-1\big{)}\Big{)}Z_1 \\
    &+
    2\bar{g}_2t\expect{P_2}X_2 +
    2\bar{g}_2t\expect{Q_2}Y_2 - \Big{(}1-2(\bar{g}_2t)^2\big{(}\expect{Q_2^2}+\expect{P_2^2}-1\big{)}\Big{)}Z_2 \\
    &+
    4\bar{g}_1\bar{g}_2t^2\expect{P_1P_2}X_1X_2 + 4\bar{g}_1\bar{g}_2t^2\expect{P_1Q_2}X_1Y_2 - 2\bar{g}_1t\expect{P_1}X_1Z_2 \\
    &+
    4\bar{g}_1\bar{g}_2t^2\expect{Q_1P_2}Y_1X_2 + 4\bar{g}_1\bar{g}_2t^2\expect{Q_1Q_2}Y_1Y_2 - 2\bar{g}_1t\expect{Q_1}Y_1Z_2 \\
    &-
    2\bar{g}_2t\expect{P_2}Z_1X_2 - 2\bar{g}_2t\expect{Q_2}Z_1Y_2 \\
    &+
    \Big{(}1-2(\bar{g}_1t)^2\big{(}\expect{Q_1^2}+\expect{P_1^2}-1\big{)}-2(\bar{g}_2t)^2\big{(}\expect{Q_2^2}+\expect{P_2^2}-1\big{)}\Big{)}Z_1Z_2\bigg{)}
\end{split}
\end{align}

\begin{align} \label{2-mode state_xy}
\begin{split}
    \hspace*{-0.4cm} \widetilde{\rho}_q\Big{(}\ketbra{+}{+}_1\otimes\ket{+i}&\bra{+i}_2\Big{)}
    = \\
    = \frac{1}{4}\bigg{(}\mathds{1} &+ \Big{(}1-2(\bar{g}_1t)^2\expect{P_1^2}\Big{)}X_1 - (\bar{g}_1t)^2\big{(}\expect{Q_1P_1}+\expect{P_1Q_1}\big{)}Y_1 + 2\bar{g}_1t\big{(}\expect{P_1} -\bar{g}_1t\big{)}Z_1 \\
    &-
    (\bar{g}_2t)^2\big{(}\expect{Q_2P_2}+\expect{P_2Q_2}\big{)}X_2 + \Big{(}1-2(\bar{g}_2t)^2\expect{Q_2^2}\Big{)}Y_2 + 2\bar{g}_2t\big{(}\expect{Q_2} -\bar{g}_2t\big{)}Z_2 \\
    &-
    (\bar{g}_2t)^2\big{(}\expect{Q_2P_2}+\expect{P_2Q_2}\big{)}X_1X_2 + \Big{(}1-2(\bar{g}_1t)^2\expect{P_1^2} -2(\bar{g}_2t)^2\expect{Q_2^2}\Big{)}X_1Y_2 \\
    &+
    2\bar{g}_2t\big{(}\expect{Q_2} -\bar{g}_2t\big{)}X_1Z_2 - (\bar{g}_1t)^2\big{(}\expect{Q_1P_1}+\expect{P_1Q_1}\big{)}Y_1Y_2 \\
    &+
    2\bar{g}_1t\big{(}\expect{P_1} -\bar{g}_1t\big{)}Z_1Y_2 +
    4\bar{g}_1\bar{g}_2t^2\expect{P_1Q_2}Z_1Z_2\bigg{)}
\end{split}
\end{align}

Reading off the Bloch vector components from these expressions, we can write 14 relations that express 14 unique two-qubit Pauli expectations in terms of the 14 Gaussian moments of the two-mode state:
\begin{equation} \label{shifted Pauli expvals}
\begin{aligned}
    \expect{X_1}_{\ket{\psi_1}} \approx 2\bar{g}_1t \expect{P_1}, \quad
    \expect{Y_1}_{\ket{\psi_1}} \approx 2\bar{g}_1t \expect{Q_1},& \quad
    \expect{X_2}_{\ket{\psi_1}} \approx 2\bar{g}_2t \expect{P_2}, \quad
    \expect{Y_2}_{\ket{\psi_1}} \approx 2\bar{g}_2t \expect{Q_2} \\
    \expect{Z_1}_{\ket{\psi_1}} + 1 + 2(\bar{g}_1t)^2 \approx 2(\bar{g}_1t)^2\Big{(} \expect{Q_1^2}+ \expect{P_1^2}\Big{)},& \quad
    \expect{Z_2}_{\ket{\psi_1}} + 1 + 2(\bar{g}_2t)^2 \approx 2(\bar{g}_2t)^2\Big{(}\expect{Q_2^2}+\expect{P_2^2}\Big{)} \\
    \expect{X_1}_{\ket{\psi_2}} - 1 \approx -2(\bar{g}_1t)^2 \expect{P_1^2},& \quad
    \expect{Y_2}_{\ket{\psi_2}} - 1 \approx -2(\bar{g}_2t)^2 \expect{Q_2^2} \\
    \expect{Y_1}_{\ket{\psi_2}} \approx -2(\bar{g}_1t)^2 \expect{\frac{Q_1P_1+P_1Q_1}{2}},& \quad
    \expect{X_2}_{\ket{\psi_2}} \approx -2(\bar{g}_2t)^2 \expect{\frac{Q_2P_2+P_2Q_2}{2}} \\
    \expect{X_1X_2}_{\ket{\psi_1}} \approx 4\bar{g}_1\bar{g}_2t^2 \expect{P_1P_2},& \quad
    \expect{X_1Y_2}_{\ket{\psi_1}} \approx 4\bar{g}_1\bar{g}_2t^2 \expect{P_1Q_2} \\
    \expect{Y_1X_2}_{\ket{\psi_1}} \approx 4\bar{g}_1\bar{g}_2t^2 \expect{Q_1P_2},& \quad
    \expect{Y_1Y_2}_{\ket{\psi_1}} \approx 4\bar{g}_1\bar{g}_2t^2 \expect{Q_1Q_2}
\end{aligned}
\end{equation}
where again $\ket{\psi_1}\equiv \ket{g}_1\otimes \ket{g}_2$ and $\ket{\psi_2}\equiv\ket{+}_1\otimes \ket{+i}_2$.
Defining shifted versions of some Pauli expectations,
\begin{align} \label{shifted Paulis}
    \overline{\expect{Z_1}}_{\ket{\psi_1}} &\equiv  \expect{Z_1}_{\ket{\psi_1}} + 1 + 2(\bar{g}_1t)^2 \nn \\
    \overline{\expect{Z_2}}_{\ket{\psi_1}} &\equiv \expect{Z_2}_{\ket{\psi_1}} + 1 + 2(\bar{g}_2t)^2 \nn \\
    \overline{\expect{X_1}}_{\ket{\psi_2}} &\equiv \expect{X_1}_{\ket{\psi_2}} - 1 \nn \\
    \overline{\expect{Y_2}}_{\ket{\psi_2}} &\equiv \expect{Y_2}_{\ket{\psi_2}} - 1,
\end{align}
allows us to express the above affine relationships as a linear relationship,
\begin{align}
    \bo \approx \mathcal{M} \bgamma,
\end{align}
where $\bo$ is a vector of Pauli expectations (some are shifted as above) and $\bgamma$ is a vector of (non-centralized) first- and second-order moments of the two-mode state.  

The explicit form of $\mathcal{M}^{-1}$ is: 
\begin{align*}
\scriptsize
    \mathcal{M}^{-1}=\begin{pmatrix}
        \frac{1}{2\bar{g}_1t} & 0 & 0 & 0 & 0 & 0 & 0 & 0 & 0 & 0 & 0 & 0 & 0 & 0\\
        0 & \frac{1}{2\bar{g}_1t} & 0 & 0 & 0 & 0 & 0 & 0 & 0 & 0 & 0 & 0 & 0 & 0\\
        0 & 0 & \frac{1}{2\bar{g}_2t} & 0 & 0 & 0 & 0 & 0 & 0 & 0 & 0 & 0 & 0 & 0\\
        0 & 0 & 0 & \frac{1}{2\bar{g}_2t} & 0 & 0 & 0 & 0 & 0 & 0 & 0 & 0 & 0 & 0\\
        0 & 0 & 0 & 0 & \frac{1}{2(\bar{g}_1t)^2} & \frac{1}{2(\bar{g}_1t)^2} & 0 & 0 & 0 & 0 & 0 & 0 & 0 & 0\\
        0 & 0 & 0 & 0 & 0 & \frac{-1}{2(\bar{g}_1t)^2} & 0 & 0 & 0 & 0 & 0 & 0 & 0 & 0\\
        0 & 0 & 0 & 0 & 0 & 0 & \frac{-1}{2(\bar{g}_2t)^2} & 0 & 0 & 0 & 0 & 0 & 0 & 0\\
        0 & 0 & 0 & 0 & 0 & 0 & \frac{1}{2(\bar{g}_2t)^2} & \frac{1}{2(\bar{g}_2t)^2} & 0 & 0 & 0 & 0 & 0 & 0\\
        0 & 0 & 0 & 0 & 0 & 0 & 0 & 0 & \frac{-1}{2(\bar{g}_1t)^2} & 0 & 0 & 0 & 0 & 0\\
        0 & 0 & 0 & 0 & 0 & 0 & 0 & 0 & 0 & \frac{-1}{2(\bar{g}_2t)^2} & 0 & 0 & 0 & 0\\
        0 & 0 & 0 & 0 & 0 & 0 & 0 & 0 & 0 & 0 & \frac{1}{4\bar{g}_1\bar{g}_2t^2} & 0 & 0 & 0\\
        0 & 0 & 0 & 0 & 0 & 0 & 0 & 0 & 0 & 0 & 0 & \frac{1}{4\bar{g}_1\bar{g}_2t^2} & 0 & 0\\
        0 & 0 & 0 & 0 & 0 & 0 & 0 & 0 & 0 & 0 & 0 & 0 & \frac{1}{4\bar{g}_1\bar{g}_2t^2} & 0\\
        0 & 0 & 0 & 0 & 0 & 0 & 0 & 0 & 0 & 0 & 0 & 0 & 0 & \frac{1}{4\bar{g}_1\bar{g}_2t^2}
    \end{pmatrix}
\end{align*}
where we have ordered the entries in the Pauli expectation and Gaussian moment vectors as
\begin{align} \label{gamma_ordering}
    {\bf{\gamma}}
    = [&\expect{Q_1},\expect{P_1},\expect{Q_2},\expect{P_2},\expect{Q_1^2},\expect{P_1^2},\expect{Q_2^2},\expect{P_2^2},\expect{(Q_1P_1+P_1Q_1)/2}, \nn\\
    &\expect{(Q_2P_2+P_2Q_2)/2},\expect{Q_1Q_2},\expect{Q_1P_2},\expect{P_1Q_2},\expect{P_1P_2}]^T
\end{align}
and
\begin{align}\label{Pauli_ordering}
\begin{split}
    {\bf{p}}=[&\expect{Y_1}_{\ket{\psi_1}},\expect{X_1}_{\ket{\psi_1}},\expect{Y_2}_{\ket{\psi_1}},\expect{X_2}_{\ket{\psi_1}},\overline{\expect{Z_1}}_{\ket{\psi_1}},\overline{\expect{X_1}}_{\ket{\psi_2}},\overline{\expect{Y_2}}_{\ket{\psi_2}},\overline{\expect{Z_2}}_{\ket{\psi_1}},\expect{Y_1}_{\ket{\psi_2}},\expect{X_2}_{\ket{\psi_2}},\\
    &\expect{Y_1Y_2}_{\ket{\psi_1}},\expect{Y_1X_2}_{\ket{\psi_1}},\expect{X_1Y_2}_{\ket{\psi_1}},\expect{X_1X_2}_{\ket{\psi_1}}]^T.
\end{split}
\end{align}

\end{document}